\documentclass{article}

\usepackage{latexsym,fullpage,amsthm,amsmath,amssymb,amstext,url,verbatim}

\newtheorem{corollary}{Corollary}

\newtheorem{definition}{Definition}
\newtheorem{theorem}[definition]{Theorem}
\newtheorem{lemma}[definition]{Lemma}
\newtheorem{claim}[definition]{Claim}

\newenvironment{proofof}[1]{\smallskip
\noindent {\bf Proof of #1.  }}{\hfill$\Box$}

\def\floor#1{\left\lfloor #1 \right\rfloor}
\def\ceil#1{\left\lceil #1 \right\rceil}

\begin{document}

\title{On Exact Algorithms for Permutation CSP\thanks{This work is partially supported by the ANR project AGAPE (ANR-09-BLAN-0159).}}

\author{Eun Jung Kim\footnote{LAMSADE, Univ. Paris-Dauphine and CNRS, France. \texttt{Email:\small{\{eun-jung.kim@dauphine.fr\}}}} and
Daniel Gon\c{c}alves\footnote{LIRMM, Univ. Montpellier 2 and CNRS, France. \texttt{Email:\small{\{daniel.goncalves@lirmm.fr\}}}} }

\date{}

\maketitle

\begin{abstract}
In the {\sc Permutation Constraint Satisfaction Problem} ({\sc Permutation
CSP}) we are given a set of variables $V$ and a set of constraints
${\mathcal C}$, in which constraints are tuples of elements of $V$. The
goal is to find a total ordering of the variables, $\pi\ : V
\rightarrow [1,\ldots,|V|]$, which satisfies as many constraints as
possible.  A constraint $(v_1,v_2,\ldots,v_k)$ is satisfied by an
ordering $\pi$ when $\pi(v_1)<\pi(v_2)<\ldots<\pi(v_k)$. An instance
has arity $k$ if all the constraints involve at most $k$ elements.

This problem expresses a variety of permutation problems including
{\sc Feedback Arc Set} and {\sc Betweenness} problems. A naive
algorithm, listing all the $n!$ permutations, requires $2^{O(n\log{n})}$
time. Interestingly, {\sc Permutation CSP} for arity
$2$ or $3$ can be solved by Held-Karp type algorithms in time
$O^*(2^n)$, but no algorithm is known for arity at least $4$ with
running time significantly better than $2^{O(n\log{n})}$. In this
paper we resolve the gap by showing that {\sc Arity 4 Permutation CSP} cannot be solved in time $2^{o(n\log{n})}$ unless ETH fails.
\end{abstract}

\section{Introduction}

Many combinatorial problems are intractable in the sense that they are unlike to have polynomial-time algorithms. One possible strategy to deal with intractability is to design {\em moderately exponential-time algorithms}. Such algorithms solve the problems optimally on any given instance. Even though exponentially many steps are required in the worst case, an exact algorithm with a slow-growing runtime function may work quite well in practice. From theoretic viewpoint, an $O^*(1.999^n)$-time algorithm is better than $O^*(2^n)$ algorithm, which is in turn better than an $O^*(n!)$-time algorithm.

The study of moderately exponential-time algorithms can be traced back to the $O^*(2^n)$-time algorithm for {\sc Hamiltonian Cycle} by Held and Karp in 1962. Ever since the $O^*(2^n)$ worst-case bound seemed to be impenetrable for almost fifty years till $O^*(1.657^n)$-time (randomized) algorithm became known by Bj\"{o}rklund in \cite{B10}. This is only a part of the success story in search for faster exact algorithms, especially in the last decade. Examples include the $O^*(2^{\omega n/3})$-time algorithm, where $\omega$ is the matrix multiplication exponent, for {\sc Max-2-CSP} \cite{W05}, a sequence of algorithmic development for {\sc Coloring} culminating in $O^*(2^n)$-time algorithm \cite{BHK09}, the very recent $O^*(c^n)$-time algorithm for {\sc Scheduling} \cite{CPPW11} for $c<2$ and lots more.


The resistance of $2^n$ barrier (and its breakdown) is repeatedly observed in combinatorial problems. While remarkable algorithmic improvement has been made, for fundamental problems such as {\sc Circuit SAT}, {\sc TSP} and {\sc Coloring}, the (asymptotic) $O^*(2^n)$ runtime remains the current best. More generally, we ask what would be the lower bounds for combinatorial problems. It is widely believed that certain NP-complete problems such as {\sc 3-Coloring}, {\sc Independent Set}, {\sc 3-Sat} are not likely to have subexponential-time algorithms. This assumption is called \texttt{Exponential Time Hypothesis (ETH)} and it serves as a common ground to prove a number of hardness results.

\newenvironment{ETH}[1][Exponential Time Hypothesis:]{\begin{trivlist}
\item[\hskip \labelsep {\bfseries #1}]}{\end{trivlist}}
\begin{ETH}
There exists an $\epsilon > 0$ such that no algorithm
solves {\sc 3-Sat} in time $2^{\epsilon n}$, where $n$ is the number
of variables.
\end{ETH}


In this paper, we are interested in a family of problems called {\sc Permutation CSP}. In the problem {\sc Permutation CSP}, we are given a set of variables $V$ and a set of constraints
${\mathcal C}$, which constraints are tuples of elements of $V$. The
goal is to find a total ordering of the variables, $\pi\ : V
\rightarrow [1,\ldots,|V|]$, which satisfies as many constraints as
possible.  A constraint $(v_1,v_2,\ldots,v_k)$ is satisfied by an
ordering $\pi$ when $\pi(v_1)<\pi(v_2)<\ldots<\pi(v_k)$. An instance
has arity $k$ if all the constraints involve at most $k$ elements. In case the arity is bounded by $k$, we call the problem {\sc Arity $k$ Permutation CSP}.

The {\sc Permutation CSP} is NP-complete even when restricted to instances of arity two, which is also known as {\sc Maximum Acyclic Subgraph} or {\sc Feedback Arc Set}. A trivial algorithm for arbitrary arity considers every possible ordering of $V$ and counts the number of satisfied constraints.
This requires $O(n!n|{\mathcal C}|)$ time, which is $2^{O(n\log{n})}$. However, we can do much better on instances with arity up to three using standard dynamic programming. For example, it is fairly straightforward to apply the framework of \cite{BFK11,KP10} and obtain $2^n$-time algorithm.


Although it is not difficult to design a dynamic programming over subsets for arity up to three, it is not clear how we can proceed with arity four and so on. Is it possible to have such an algorithm, or {\em any} algorithm of runtime $O^*(c^n)$ for some constant $c$? We answer this question in the negative. The following theorem summarizes our main result stating that such improvement is impossible under ETH.

\begin{theorem}
\label{th-main}
Assuming ETH, there is no $2^{o(n\log{n})}$-algorithm for {\sc Arity 4 Permutation CSP} (and thus for {\sc Arity $k$ Permutation CSP}, $k\ge 4$).
\end{theorem}

Our result is built on two previous results: Impagliazzo et al.~\cite{IPZ01} and Lokshtanov et al.~\cite{LMS11}. In \cite{LMS11}, the authors prove a computational lower bound for the {\sc $n\times n$ Clique} problem and transfer the lower bound to those of other natural combinatorial problems. The {\sc $n\times n$ Clique} is designed so that an improvement over the brute-force search would contradict ETH. It is not difficult to build a reduction from {\sc $n\times n$ Clique} to {\sc Permutation CSP} with arity six, thus showing Theorem \ref{th-main} holds for arity $k\ge 6$. The technical difficulty arises when we try to get down the arity down to four. To do this, we resort to the Sparsification Lemma of \cite{IPZ01} and construct a variation of {\sc $n\times n$ Clique} with strictly constrained properties.

To be more specific, we give a sequence of reductions starting from an instance of {\sc 3-Sat} in which the maximum frequency (i.e. the number of clauses containing a variable) is bounded by a fixed constant $f$. Such a {\sc 3-Sat} instance is reduced to a {\sc 3-Coloring} of degree bounded by $f'$. From {\sc 3-Coloring}, we construct an instance of {\sc $n\times n$ Clique} along the line of \cite{LMS11}. In order to construct an {\sc $n\times n$ Clique} instance constrained in a subtle manner, we use Brooks' theorem and ternary grey code. Finally from such an instance of {\sc $n\times n$ Clique}, we reduce to {\sc Arity 4 Permutation CSP}. The whole chain of reductions are designed so that an $O^*(2^{o(n\log{n})})$-time algorithm for {\sc Arity 4 Permutation CSP} implies that we can solve {\sc 3-Sat} instance of frequency $f$ in time $O(2^{o(n)})$ for {\em any given} $f$, which is unlikely.


An interesting dichotomy is observed as a corollary of our main result.

\begin{corollary}
The problem {\sc Arity $k$ Permutation CSP} can be solved in time $O^*(2^n)$ if $k\leq 3$. Otherwise, such an algorithm is unlikely to exist under ETH.
\end{corollary}

In the next section we give some definitions and a simple proof of the
fact that {\sc Arity 6 Permutation CSP} has no
$2^{o(n\log{n})}$-algorithm, unless ETH fails. Then we prove
Theorem~\ref{th-main} in Section~\ref{sec:arity4}.

\section{Warm-up}
\label{sec:preliminary}

Let us define the {\sc $n\times n$ Clique} Problem introduced by
Lokshtanov et al.~\cite{LMS11}.

\begin{quote}
{\sc $n\times n$ Clique}\\
{\bf Input:} a graph $G$ with vertex set $V(G) = [n]\times [n]$. (The vertex $(i,j)$ is said to be in
row $i$ and column $j$.) \\
{\bf Goal:} Determine if there exists an $n$-clique in $G$ with exactly one element from each row.
\end{quote}

\begin{theorem}[Lokshtanov et al.~\cite{LMS11}]
\label{ETHclique}
Assuming ETH, there is no $2^{o(n\log{n})}$-time algorithm for {\sc
  $n\times n$ Clique}.
\end{theorem}

Using this result we prove the following.
\begin{theorem}
\label{th-arity6}
Assuming ETH, there is no $2^{o(n\log{n})}$-algorithm for {\sc Arity 6 Permutation CSP}.
\end{theorem}
\begin{proof}
To prove that, we show that one could use a
$2^{o(n\log{n})}$-algorithm for {\sc Arity 6 Permutation CSP} to
design a $2^{o(n\log{n})}$-time algorithm for {\sc $n\times n$
  Clique}, which does not exist if we assume ETH (by
Theorem~\ref{ETHclique}). This is due to an easy reduction from
{\sc $n\times n$ Clique} to {\sc Arity 6 Permutation CSP}. Note that in order to achieve the desired lower bound, the reduction needs to produce an {\sc
  Arity 6 Permutation CSP} instance $(V,\mathcal{C})$ with $|V|=O(n)$.

Given an instance $G$ of {\sc $n\times n$ Clique}, we build an instance $(V,\mathcal{C})$ containing $4n+1$ elements in $V$ as follows. There are (a) $n$ elements $r_i$, $i\in[n]$, corresponding to the rows of $G$, (b) $n$ elements $c_i$, $i\in[n+1]$, corresponding to the columns of $G$ (except $c_{n+1}$ that does not exactly correspond to a column), and (c) $2n$ 'dummy' elements $d_i$, $i\in[2n]$. The constraint set $\mathcal{C}$ is
the union of three types of constraints, the constraints $\mathcal{C}_G$ depending $G$ and the structural constraints
$\mathcal{C}^1_S$ and $\mathcal{C}^2_S$ that force the optimal
orderings to be of the form $d_1d_2\ldots d_{2n}c_1R_1c_2R_2\ldots
c_nR_nc_{n+1}$, where each $R_i$ is a (possibly empty) sequence of
$r_i$'s.

Formally, $\mathcal{C}^1_S =
\{(d_a,d_b,d_c,d_d,c_j,c_{j'})\ |\ \forall \ 1\le a<b<c<d\le 2n \text{
  and } \forall \ 1\le j<j'\le n+1 \}$, $\mathcal{C}^2_S =
\{(c_1,r_i,c_{n+1})\ |\ \forall i\in[n]\}$. $\mathcal{C}_G$ is such
that for every edge $(i,j)(i',j')\in E(G)$ there is a constraint
requiring $r_i$ and $r_{i'}$ to be respectively between $c_j$ and
$c_{j+1}$, and between $c_{j'}$ and $c_{j'+1}$. Since we avoid an element appearing more than once in a constraint, a constrain in $\mathcal{C}_G$ can be one of three types according to the value $j'-j$. $\mathcal{C}_G =
\{(c_j,r_i,c_{j+1},c_{j'},r_{i'},c_{j'+1})\ | \ \forall
(i,j)(i',j')\in E(G)\ \text{with} \ j+2\le j'\} \cup
\{(c_j,r_i,c_{j+1},r_{i'},c_{j+2})\ | \ \forall (i,j)(i',j+1)\in
E(G)\} \cup \{(c_j,r_i,r_{i'},c_{j+1})\ | \ \forall (i,j)(i',j)\in
E(G)\ \text{with} \ i<i'\}$.  Note that since we ask for one element
per row, we can assume that $G$ has no edge $(i,j)(i,j')$. Thus
$\mathcal{C}_G$ has exactly one constraint per edge.

\begin{claim}
Any optimal ordering is of the form $d_1d_2\ldots
d_{2n}c_1R_1c_2R_2\ldots c_nR_nc_{n+1}$, where each $R_i$ is a
(possibly empty) sequence of $r_i$'s. Thus in any optimal ordering all
the constraints of $\mathcal{C}^1_S\cup\mathcal{C}^2_S$ are satisfied.
\end{claim}
\begin{proof}
Notice that sequences of this type exist and that they satisfy all the
constraints in $\mathcal{C}^1_S\cup\mathcal{C}^2_S$. Note also that
since $G$ has at most ${n^2}\choose{2}$ edges,
$|\mathcal{C}_G|\le$${n^2}\choose{2}$.

If in some ordering two elements $d_i$ and $d_j$ are misplaced with
respect to each other, then the ${2n-2}\choose{2}$${n+1}\choose{2}$
constraints of $\mathcal{C}^1_S$ involving them are unsatisfied, and
this cannot be compensated by the $|\mathcal{C}_G|\le$${n^2}\choose{2}$
constraints in $\mathcal{C}_G$ (for a sufficiently large $n$).

Similarly, if in some ordering two elements $d_i$ and $c_j$ are
misplaced with respect to each other, then the ${2n-1}\choose{3}$$n$
constraints of $\mathcal{C}^1_S$ involving them are unsatisfied, and
this cannot be compensated by the $|\mathcal{C}_G|\le$${n^2}\choose{2}$
constraints in $\mathcal{C}_G$ (for a sufficiently large $n$).

Similarly again, if in some ordering two elements $c_i$ and $c_j$ are
misplaced with respect to each other, then the ${2n}\choose{4}$
constraints of $\mathcal{C}^1_S$ involving them are unsatisfied, and
this cannot be compensated by the $|\mathcal{C}_G|\le$${n^2}\choose{2}$
constraints in $\mathcal{C}_G$ (for a sufficiently large $n$).

Finally, if some $r_i$ lies before $c_1$ or after $c_{n+1}$, it cannot
be part of any satisfied constraint, contrary to the case where $r_i$
lies in between $c_1$ and $c_{n+1}$. In this case, at least
one satisfied constraint involves $r_i$, the one in $\mathcal{C}^2_S$.
Thus the optimal orderings are of the desired form.
\end{proof}

\begin{claim}
An optimal ordering satisfies at most ${n}\choose{2}$ constraints from
$\mathcal{C}_G$, and equality holds if and only if $G$ has an
$n$-clique with one element from each row.
\end{claim}
\begin{proof}
In an optimal ordering, for every $r_i$ there is a unique value,
$\phi(i)$, such that $r_i\in R_{\phi(i)}$ (i.e. $r_i$ lies in between
$c_{\phi(i)}$ and $c_{\phi(i)+1}$).  Thus, given a pair of row
elements, $r_i$ and $r_{i'}$, an optimal ordering satisfies at most
one constraint involving both elements (for example the constraint
$(c_{\phi(i)},r_i,c_{\phi(i)+1},c_{\phi(i')},r_{i'},c_{\phi(i')+1})$
if $\phi(i)+2\le \phi(i')$). This implies that a constraint involving
$r_i$ and $r_{i'}$ is satisfied if and only if
$(i,\phi(i))(i',\phi(i')) \in E(G)$. Thus, ${n}\choose{2}$ constraints
from $\mathcal{C}_G$ are satisfied only if the vertices $(i,\phi(i))$
with $i\in [n]$ form a $n$-clique with one vertex per row. Conversely,
if $G$ has a $n$-clique with one vertex per row, it is easy to order
the elements in such a way that ${n}\choose{2}$ constraints from
$\mathcal{C}_G$ are satisfied. Note that if both $(i,j)$ and $(i',j)$,
with $i<i'$, are in the $n$-clique one should put $r_i$ before
$r_{i'}$ (which are both between $c_j$ and $c_{j+1}$) to satisfy the
constraint $(c_j,r_i,r_{i'},c_{j+1})$ associated to the edge
$(i,j)(i',j)$. This concludes the proof of the claim.
\end{proof}
Thus an optimal ordering satifies
$|\mathcal{C}^1_S|+|\mathcal{C}^2_S|+ {{n}\choose{2}}$ constraints if and
only if $G$ has an $n$-clique with one element from each row, and this
concludes the proof of the theorem.
\end{proof}

\section{Main result}
\label{sec:arity4}

We have shown that a trivial enumeration algorithm for {\sc Arity 6 Permutation CSP} cannot be significantly improved under ETH by a reduction from {\sc $n\times n$ Clique}. In this section, we shall generalize this result to instances of arity four. For this, we successively establish lower bounds on several problems, and finally on {\sc Arity 4 Permutation CSP} assuming ETH.

In~\cite{LMS11}, the authors prove Theorem~\ref{ETHclique} by a
reduction from {\sc 3-Coloring} to {\sc $n\times n$ Clique} such that a $2^{o(n\log{n})}$-algorithm for {\sc $n\times n$ Clique}
implies $2^{o(n)}$-algorithm for {\sc 3-Coloring}. We follow the same
line of reduction, but we need to constrain {\sc $n\times n$ Clique} in a careful way so that we can finally reduce to {\sc Arity 4 Permutation CSP} instances.

In the following, we present algorithmic lower bounds on {\sc $f$-Sparse 3-Sat}, {\sc $f'$-Sparse 3-Coloring}, {\sc $D$-Degree Constrained $n\times n$ Clique}, and {\sc $D$-Degree   Constrained $2n\times 2n$ Biclique}. Then we prove
Theorem~\ref{th-main} by an appropriate reduction of {\sc $D$-Degree
  Constrained $2n\times 2n$ Biclique} to {\sc Arity 4 Permutation CSP}.

\begin{quote}
{\sc $f$-Sparse 3-Sat}\\
{\bf Input:} A 3CNF formula with $n$ variables and $m$ clauses in which each variable is contained in at most $f$ clauses.\\
{\bf Goal:} Determine if there is a satisfying assignment.
\end{quote}

\begin{quote}
{\sc $f'$-Sparse 3-Coloring}\\
{\bf Input:} A graph $G$ on $n$ vertices in which the maximum degree is at most $f'$.\\
{\bf Goal:} Determine if $G$ is 3-colorable.
\end{quote}


The Sparsification Lemma by Impagliazzo et al. states that for every $\epsilon >0$, 3CNF formula on $n$ variables can be converted as a disjunction of at most $2^{\epsilon n}$ 3CNF such that in each 3CNF, the frequency bounded by a function depending only on $\epsilon$ (see~\cite{IPZ01}, Corollary 1). Moreover, this disjunction can be constructed in time $2^{\epsilon n}poly(n)$. The following is a direct consequence of it.

\begin{theorem}[Impagliazzo et al.~\cite{IPZ01}]\label{ETHs3s}
Assuming ETH, {\sc $f$-Sparse 3-Sat} cannot be solved in $2^{o(n)}$-time for every fixed $f>0$.
\end{theorem}

Notice that the above theorem does not exclude the possibility of subexponential-time algorithm for {\em some} $f$. The following theorem follows from a slight modification of the well-known reduction \cite{S96}.

\begin{lemma}\label{ETHb3c}
Assuming ETH, {\sc $f'$-Sparse 3-Coloring} cannot be solved in $2^{o(n)}$-time for every fixed $f'>0$.
\end{lemma}
\begin{proof}
We give a sketch of the polynomial-time reduction in \cite{S96} and point out how we modify the reduction in order to ensure the maximum degree. The reduction in \cite{S96} constructs a graph $G$ from a given 3CNF formula $\Phi$ as follows. The vertex set $V(G)$ contains (a) three vertices $T,F,B$ forming a triangle, (b) two literal vertices $v_i$ and $\bar{v}_i$ corresponding to each variable $x_i$ of $\Phi$, (c) an \texttt{OR}-gadget $C_j$ corresponding to $j$-th clause consisting of six vertices. Apart from the triangle on $T,F,B$, edges inside \texttt{OR}-gadget, the edge set $E(G)$ additionally connects (i) the pairs $v_i$ and $\bar{v}_i$, (ii) every literal vertex $v_i$/$\bar{v}_i$ with $N$, (iii) the literal vertex $v_i$/$\bar{v}_i$ with (a vertex from) \texttt{OR}-gadget $C_j$ whenever $v_i$/$\bar{v}_i$ appears in $C_j$, (iv) the vertex $out_j$ from \texttt{OR}-gadget $C_j$ with $N$ and $F$.

The graph $G$ (in particular, the \texttt{OR}-gadget) is designed so that $\Phi$ is satisfiable if and only if $G$ is 3-colorable. Recall that the vertices $T,F,N$ forms a triangle and thus has distinct colors in any 3-coloring. We say a vertex is assigned $T$ ($F,N$ respectively) if the vertex shares the same color with $T$ ($F,N$ respectively). Essentially, two properties of $G$ ensure this if-and-only-if relation. First, in any 3-coloring of $G$, if $v_i$ is assigned $T$ then $\bar{v}_i$ is assigned $F$ and vice versa. This is due to the connections (i), (ii). Second, the vertex $out_j$ of \texttt{OR}-gadget $C_j$ is assigned $T$ due to (iv) and this, together with the design of \texttt{OR}-gadget, enforces that at least one of the literal vertices connected to $C_j$ by (iii) is assigned $T$.

The connection (iii) does not lead to unbounded degree of a literal vertex $v_i$/$\bar{v}_i$ if we reduce from {\sc $f$-Sparse 3-Sat}. Observe that unbounded degree may occur due to the connections (ii) and (iv). We can resolve this case by 'expanding' the triangle on $T,F,B$ into a triangulated ladder as long as necessary. Now the connections in (ii) and (iv) are modified so that (ii') every literal vertex $v_i$/$\bar{v}_i$ is connected with distinct $N$ vertex in the triangulated ladder, and (iv') every vertex $out_j$ from \texttt{OR}-gadget $C_j$ is connected with distinct $N$ and $F$. Note that in the modified construction, the number of vertices created are still $O(n+m)$. The maximum degree is now bounded by $\max(f+2,5)$.

Suppose {\sc $f'$-Sparse 3-Coloring} can be solved in time subexponential in the number of vertices for every fixed $f'>0$. Then, given an instance of {\sc $f$-Sparse 3-Sat} for any fixed $f$, we run the presented reduction and construct {\sc $f'$-Sparse 3-Coloring} instance with $O(n+m)$ vertices. As the obtained instance can be solved in $2^{o(n+m)}$-time and $O(m)=O(n)$, we can solve the initial {\sc $f$-Sparse 3-Sat} instance in time $2^{o(n)}+poly(n)$-time, a contradiction to Theorem \ref{ETHs3s}. Hence the statement follows.
%
\end{proof}

Let us consider an instance $G$ of {\sc $n\times n$ Clique} on the vertex set $V=[n]\times [n]$. Let $x$ be a
vertex of $G$ and $X$ is a subset of $V$. We denote the number of
edges between $x$ and $X$ by $\deg(x,X)$. The set of vertices in the
$i$-th row is denoted as $\mathcal{R}_i=\{(i,j):j\in [n]\}$.  Then
given two vertex sets $X$ and $Y$, $E(X,Y)$ is the set of edges with
one end in $X$ and one end in $Y$, and $\deg(X,Y)=|E(X,Y)|$. These
definition naturally extend to the case where $X$ or $Y$ is a single
vertex.

\begin{quote}
{\sc $D$-Degree Constrained $n\times n$ Clique} ({\sc $D$-DCnnC})\\
{\bf Input:} An instance $G$ of {\sc $n\times n$ Clique} with two additional conditions.
\begin{itemize}
\item[(A)] For every pair of rows $i,k \in [n]$, we have $\deg((i,j),\mathcal{R}_k) = \Delta^{ik}$ for every $j \in [n]$, for some constant $\Delta^{ik}$.
\item[(B)] For every vertex $(i,j)$ with $i\in [n], j\in [n-1]$, there exists a set of rows $\exists I_{i,j}\subsetneq [n]$ with $|I_{i,j}|\ge n-D$ such that $N(i,j)\cap\mathcal{R}_k = N(i,j+1) \cap \mathcal{R}_k$ for every $k\in I_{i,j}$.
\end{itemize}
{\bf Goal:} Determine if there is an $n$-clique in $G$ with exactly one element from each row.
\end{quote}

\begin{lemma}\label{ETH-C}
Assuming ETH, {\sc $D$-DCnnC} cannot be solved in $2^{o(n\log{n})}$-time for every fixed $D>0$.
\end{lemma}
\begin{proof}
Similarly to the proof of Theorem~\ref{ETHclique} in~\cite{LMS11}, the
theorem follows from a reduction of {\sc $f'$-Sparse 3-Coloring}
to {\sc $D$-DCnnC}. This reduction is such that an instance $G$ of
{\sc $f'$-Sparse 3-Coloring} with $n$ vertices is equivalent to
an instance $G'$ of {\sc $D$-DCnnC} with $D=f'$ and vertex set
$V'=[n']\times [n']$, with $n'\log(n')=O(n)$.

Let $x$ be the smallest integer such that
$(f'^2+1)+\floor{\frac{n-f'^2-1}{x}} \le 3^x$.  Now let
$n'=3^x$.  Note that since $f'$ is a constant we have $n < x3^x <
4n$ (i.e. $n < n'\log(n') < 4n$) for $n$ sufficiently large, thus
$x3^x = O(n)$ (i.e. $n'\log(n')=O(n)$).

Given $G$, partition its vertices into $f'^2+1$ parts,
$V_0,\ldots,V_{f'^2}$ such that any two vertices in the same part
are at distance at least $3$ in $G$. This is possible by Brooks'
Theorem and the fact that the maximum degree of $G^2$ is upper bounded by $f'^2$.  Then partition
each $V_i$ into $\ceil{\frac{|V_i|}{x}}$ subsets $V_i^j$, such that
$|V_i^1|\le x$ and such that all the remaining sets $V_i^j$ have size
exactly $x$. Here all the subsets have size $x$ except at most
$f'^2+1$ subsets of size at least 1 (and at most $x-1$). Thus now
$V$ is partitioned in at most $(f'^2+1) +
\floor{\frac{n-f'^2-1}{x}}$ (which is $\le n'$) subsets of size at
most $x$. Rename those subsets $X_1,X_2,\ldots$ and if necessary add
some empty sets at the end in order to have exactly $n'$ sets
$X_i$. Then add $xn'-n$ isolated vertices in $G$ and dispatch them in
the $X_i$'s in such a way that all these sets have size exactly $x$.

Now let $G'$ be the graph on $[n']\times [n']$ such that each vertex
$(i,j)$ corresponds to a 3-coloring of $X_i$. There is enough space in
each row since the stable set $X_i$ has exactly $3^{|X_i|} = 3^x = n'$
3-colorings. We enumerate these colorings in such a way that for every
$i\in [n']$ and $j\in [n'-1]$ the coloring corresponding to $(i,j)$
differs on exactly one vertex with the coloring corresponding to
$(i,j+1)$. This can be done using a ternary Gray code on $|X_i|=x$
digits, which k$^{th}$ digit corresponds to the color of the k$^{th}$
vertex in $X_i$. In $G'$ there is an edge between the vertices $(i,j)$
and $(i',j')$ if and only if the corresponding 3-colorings of $X_i$
and $X_{i'}$ are compatible, this when $i\neq i'$ and when there is no
adjacent vertices $u\in X_i$ and $v\in X_{i'}$ with the same color.

It is clear that $G$ is 3-colorable if and only if $G'$ has an
$n'$-clique with one element per row. So it remains now to check that
$G'$ verifies conditions (A) and (B). By the construction of the sets
$X_i$ (sub-partitioning the sets $V_j$), for any $i, i'\in [n]$ the
induced graph $G[X_i\cup X_{i'}]$ is a matching (generally a
non-perfect one). If there are $m$ edges in this matching one easily
sees that any 3-coloring of $X_i$ is compatible with exactly
$2^m3^{x-m}$ 3-colorings of $X_{i'}$, and thus (A) holds.
Furthermore, since two consecutive 3-colorings of $X_i$ (say the ones
corresponding to $(i,j)$ and $(i,j+1)$) differ on exactly one vertex,
say $u$, of degree at most $f'$ in $G$ there are at least
$n'-f' = n'-D$ sets $X_{k}$ without any neighbor of $u$
(eventually $X_k=X_i$).  Since $u$ is isolated in $G[X_i,X_k]$ its
color does not really matters and we clearly have that
$N(i,j)\cap \mathcal{R}_k = N(i,j+1)\cap \mathcal{R}_k$ (if $X_k=X_i$,
then $\forall j,\ N(i,j)\cap \mathcal{R}_k=\emptyset$), and thus that
(B) holds. This concludes the proof of the theorem.
\end{proof}

\begin{quote}
{\sc $D$-Degree Constrained $2n\times 2n$ Biclique} ({\sc $D$-DCnnB})\\
{\bf Input:} A graph $G$ with vertex set $V(G) = [2n]\times [2n]$ with three additional conditions.
\begin{itemize}
\item[(A)] For every edge $(i,j)(i',j')\in E(H)$ we have $1\le i\le n <i'\le 2n$ and $1\le j\le n <j'\le 2n$. Furthermore,   $(i,j)(n+i',n+j')\in E(H)$ if and only if $(i',j')(n+i,n+j)\in E(H)$.
\item[(B)] For every pair of rows $i \in [n]$ and $k\in [n+1,2n]$, we have $\deg((i,j),\mathcal{R}_k) = \Delta^{ik}$ for every $j \in [n]$, for some constant $\Delta^{ik}$.
\item[(C)] For every vertex $(i,j)$ with $i\in [n], j\in [n-1]$, there exists a set of rows $\exists I_{i,j}\subsetneq [n+1,2n]$ with $|I_{i,j}|\ge n-D$ such that $N(i,j)\cap\mathcal{R}_k = N(i,j+1) \cap \mathcal{R}_k$ for every $k\in I_{i,j}$.
\end{itemize}
{\bf Goal:} Determine if there is $K_{n,n}$ with exactly one vertex per row.
\end{quote}

\begin{lemma}\label{ETHbiclique}
Assuming ETH, {\sc $D$-DCnnB} cannot be solved in $2^{o(n\log{n})}$-time for every fixed $D>0$.
\end{lemma}
\begin{proof}
The theorem follows from a simple reduction of {\sc $D$-DCnnC} to {\sc
  $D$-DCnnB}. This reduction is such that an instance $G$ of {\sc
  $D$-DCnnC} with vertex set $[n]\times [n]$ is equivalent to an
instance $H$ of {\sc $D$-DCnnB} with vertex set $[2n]\times [2n]$.
The graph $H$ is such that $(i,j)(n+i',n+j') \in E(H)$ if and only if
$(i,j)(i',j')\in E(G)$ or if $i=i'$ and $j=j'$. It is easy to see that
if $H$ contains a $K_{n,n}$ with one vertex per row, the selected
vertices in the first (resp. last) $n$ rows form a stable
set. Furthermore for any $i\in [n]$, according to the adjacencies
between $\mathcal{R}_i$ and $\mathcal{R}_{n+i}$, a vertex $(i,j)$ is
selected if and only if $(n+i,n+j)$ is also selected. Then it is
simple to conclude that $G$ has a clique with one vertex per row if
and only if $H$ has a $K_{n,n}$ with one vertex per row. Thus the
theorem directly follows from Lemma~\ref{ETH-C}.
\end{proof}

We can finally prove the main result.

\begin{proofof}{Theorem~\ref{th-main}}
To prove the theorem we exhibit a reduction from {\sc $D$-DCnnB} to
{\sc Arity 4 Permutation CSP}. This reduction associates each instance
$H$ (a graph on $[2n]\times [2n]$) of {\sc $D$-DCnnB} to an instance
of $(V,\mathcal{C})$ of {\sc Arity 4 Permutation CSP} with
$|V|=(2D+4)n+1$, in such a way that $H$ is a positive instance if and
only if the optimal solution to $(V,\mathcal{C})$ satisfies
${2Dn}\choose{2}$${2n+1}\choose{2}$$+(n+2)\sum_{i\in [n],i'\in
  [n+1,2n]}\Delta^{i,i'} + n^2$ constraints. By Lemma~\ref{ETHbiclique}, such a reduction clearly implies the
theorem.

As in the proof of Theorem~\ref{th-arity6}, $V$ has 3 types of elements,
 $2n$ elements $r_i$, with
$i\in[2n]$, corresponding to the rows of $H$, $2n+1$ elements $c_i$,
with $i\in[2n+1]$, corresponding to the columns of $H$ (except
$c_{n+1}$ that does not exactly correspond to a column), and $2Dn$
dummy elements $d_i$, with $i\in[2Dn]$.

The constraint set $\mathcal{C}$ is the union of several types of
constraints, the structural constraints $\mathcal{C}_S$ that force the
shape of optimal orderings, and the constraints depending on
$H$, $\mathcal{C}_H = \mathcal{C}_H^{crcr} \cup \mathcal{C}_H^{crrc}
\cup \mathcal{C}_H^{rcrc} \cup \mathcal{C}_H^{rccr}$.

Formally:
\begin{itemize}
\item $\mathcal{C}_S = \{(d_a,d_b,c_j,c_{j'})\ |\ \forall a<b
  \text{ and } \forall j<j' \}$
\item $\mathcal{C}_H^{crcr} = \{(c_j,r_i,c_{n+j'},r_{n+i'})\ |
  \ \forall \ (i,j)(n+i',n+j')\in E(H)\}$.
\item $\mathcal{C}_H^{crrc} = \{(c_j,r_i,r_{n+i'},c_{n+j'+1})\ |
  \ \forall \ (i,j)(n+i',n+j')\in E(H)\}$.
\item $\mathcal{C}_H^{rcrc} = \{(r_i,c_{j+1},r_{n+i'},c_{n+j'+1})\ |
  \ \forall \ (i,j)(n+i',n+j')\in E(H)\}$.
\item $\mathcal{C}_H^{rccr} = \{(r_i,c_{j+1},c_{n+j'},r_{n+i'})\ |
  \ \forall \ (i,j)(n+i',n+j')\in E(H)\}$.
\end{itemize}

There are some inconsistencies in this last case, we cannot have the
same element appearing several times in some constraint. This is the
case when $j=n$ and $j'=1$ (since $c_{j+1} = c_{n+j'}$), and in such a
case we replace the constraint by $(d_1,r_i,c_{n+1},r_{n+i'})$.

\begin{claim}\label{l-<n<}
There is an optimal ordering of the form $d_1d_2\ldots d_{2Dn}V_{\le
  n}V_{>n}$, where $V_{\le n}$ (resp. $V_{>n}$) is a sequence of
$r_i$'s and $c_j$' with $i$ and $j\le n$ (resp. $i$ and $j> n$).
\end{claim}
\begin{proof}
Consider any optimal ordering. Moving successively the elements $d_i$
(from $d_1$ to $d_{2Dn}$) to the i$^{th}$ position, preserves the
allready satisfied constraints, so after those moves the order remains
optimal. Then if there are two consecutive elements $v_j$ and $v_i$
(in this order), where $v_i$ is an element $r_i$ or $c_i$ with $i\le
n$, and where $v_j$ is an element $r_j$ or $c_j$ with $j> n$, since
there is no constraint involving $v_j$ and $v_{i}$ in this order,
switching those elements preserves the allready satisfied
constraints. Thus there is an optimal ordering of the desired form.
\end{proof}

\begin{claim}\label{l-dc}
There are optimal orderings of the form $d_1d_2\ldots
d_{2Dn}R_0c_1R_1c_2R_2\ldots c_{2n}R_{2n}c_{2n+1}R_{2n+1}$, where
each $R_i$ with $i\le n$ (resp. $i>n$) is a possibly empty sequence of
$r_j$ with $j\le n$ (resp. $j>n$).
\end{claim}
\begin{proof}
By Claim~\ref{l-<n<} we can consider an optimal ordering the form
$DV_{\le n}V_{>n}$. So it remains to prove that in $V_{\le n}$
(resp. $V_{>n}$) the $c_j$'s can be reordered increasingly. If $V_{\le
  n}$ has two elements $c_{j_2}$ and $c_{j_1}$ in this order, with
$1\le j_1<j_2\le n$, separated by a (possibly empty) sequence $R$ of
$r_i$'s, the following two claims show that moving $c_{j_2}$ right after $c_{j_1}$ does not decrease the number of satisfied constraints.
Actually, switching $c_{j_2}$ with the $r_i$'s in $R$ decreases the
number of satisfied constraints, but this is compensated by the last
switch between $c_{j_2}$ and $c_{j_1}$.

\begin{claim}\label{cl-cr}
Given any ordering with two consecutive elements $c_{j_2}$ and $r_i$
(in this order), with $j_2\in [2,n]$ and $i\in [n]$, switching them
decreases the number of satisfied constraints by at most $2Dn$.
\end{claim}
\begin{proof}
Let us call the constraints newly satisfied after the switch {\it activated constraints}, and the newly unsatisfied constraints {\it inactivated
  constraints}.  Note that new and inactivated constraints contain $c_{j_2}$
and $r_i$ as consecutive elements (but with opposite orders).  The
inactivated constraints are of the form $(c_{j_2},r_i,c_{j'},r_{i'})\in
\mathcal{C}_H^{crcr}$, for the values $i'$ and $j'$ such that
$(i,j)(i',j')\in E(H)$, or of the form $(c_{j_2},r_i,r_{i'},c_{j'})\in
\mathcal{C}_H^{crrc}$, for the values $i'$ and $j'$ such that
$(i,j)(i',j'-1)\in E(H)$. Similarly the new constraints are of the
form $(r_i,c_{j_2},c_{j'},r_{i'})\in \mathcal{C}_H^{rccr}$, for the
values $i'$ and $j'$ such that $(i,j-1)(i',j')\in E(H)$, or of the
form $(r_i,c_{j_2},r_{i'},c_{j'})\in \mathcal{C}_H^{rcrc}$, for the
values $i'$ and $j'$ such that $(i,j-1)(i',j'-1)\in E(H)$.

Condition (C) implies that the neighbors $N(i,j_2)$ and $N(i,j_2-1)$ are similar in some sense. To be precise, consider a vertex $(i',j')\in V(H)$ with $i',j'\in [n+1,2n]$. If the vertex is picked from row $i' \in I_{i,j_2-1}$ as defined in condition (C), then we have $(i,j_2)(i',j')\in E(H)$ if and only if $(i,j_2-1)(i',j')\in E(H)$. Hence it follows from the construction of $\mathcal{C}_H$ that $$(c_{j_2},r_i,c_{j'},r_{i'}) \in \mathcal{C}_H^{crcr} \text{ if and only if }(r_i,c_{j_2},c_{j'},r_{i'}) \in \mathcal{C}_H^{rccr}$$ and $$(c_{j_2},r_i,r_{i'},c_{j'+1})\in \mathcal{C}_H^{crrc} \text{ if and only if } (r_i,c_{j_2},r_{i'},c_{j'+1})\in \mathcal{C}_H^{rcrc}$$

Therefore, whenever the constraint $(c_{j_2},r_i,c_{j'},r_{i'})\in \mathcal{C}_H^{crcr}$
(resp. $(c_{j_2},r_i,r_{i'},c_{j'+1})\in \mathcal{C}_H^{crrc}$) becomes inactivated after the switch, the
constraint $(r_i,c_{j_2},c_{j'},r_{i'})\in \mathcal{C}_H^{rccr}$
(resp. $(r_i,c_{j_2},r_{i'},c_{j'+1})\in \mathcal{C}_H^{rcrc}$) becomes activated.

If the vertex $(i',j')\in V(H)$ is picked from row $i' \notin I_{i,j_2-1}$, then the inactivated constraints may not be compensated by activated constraints as in the case $i' \in I_{i,j_2-1}$. Hence the cost of the switch is bounded by $\sum_{i'\notin I_{i,j_2-1}} 2|N_H(i,j_2)\cap \mathcal{R}_{i'}|\leq 2Dn$.
\end{proof}

\begin{claim}\label{cl-cc}
Given an ordering starting by $d_1,d_2,\ldots,d_{2Dn}$, if there are
two consecutive $c_{j_2}$ and $c_{j_1}$ (in this order) with $j_1 <
j_2$, switching them increases the number of satisfied constraints by
at least ${2Dn}\choose{2}$.
\end{claim}
\begin{proof}
Indeed, there is no constraint involving $c_{j_2}$ and $c_{j_1}$ in
this order. Thus switching them only adds new satisfied constraints,
including the ${2Dn}\choose{2}$ constraints
$(d_a,d_b,c_{j_1},c_{j_2}) \in \mathcal{C}_S$.
\end{proof}

Thus moving $c_{j_2}$ along $R$ (which has length at most $n$) costs
at most $n(2Dn)$ which is compensated by the benefits from the last
switch between $c_{j_2}$ and $c_{j_1}$ (${2Dn}\choose{2}$$\ge
2Dn^2$). Similar arguments hold for $V_{>n}$, and this concludes the
proof of the lemma.
\end{proof}

A {\it convenient} ordering is an orderings of the form $d_1d_2\ldots
d_{2Dn}c_1R_1c_2R_2\ldots c_{2n}R_{2n}c_{2n+1}$, where each $R_i$ with
$i\le n$ (resp. $i>n$) is a possibly empty sequence of $r_j$ with
$j\le n$ (resp. $j>n$).  Given a convenient ordering, one can define a
function $\phi :[2n]\longrightarrow [2n]$ such that $\forall i\in
[2n]$, $r_i\in R_{\phi(i)}$. Given such a function one can define a
set of vertices with exactly one vertex per row, $V_\phi =
\{(i,\phi(i))\ |\ i\in [2n]\}$.

\begin{claim}\label{l-sequ}
There are optimal orderings that are convenient.
\end{claim}
\begin{proof}
Consider an optimal ordering as described in Claim~\ref{l-dc}.  If
there was a $r_i$ just before $c_1$ (resp. just after $c_{2n+1}$),
switching those two elements would not contradict any satisfied
constraint: there is no constraint with a $r_i$ before $c_1$ or after
$c_{2n+1}$.  Thus there exists an optimal orderings of the desired
form.
\end{proof}

In the following let us denote $\sum \Delta^{i,i'}$ the following
sum $\sum_{1\le i\le n <i'\le 2n}\Delta^{i,i'}$.
\begin{claim}\label{l-nb}
The number of constraints satisfied by a convenient ordering is
${2Dn}\choose{2}$${2n+1}\choose{2}$$+(n+2)\sum\Delta^{i,i'} +
|E(H[V_\phi])|$.
\end{claim}

\begin{proof}
It is clear that a convenient ordering satisfies all the
${2Dn}\choose{2}$ constraints in $\mathcal{C}_S$. For the satisfied
constraints in $\mathcal{C}_H$, consider any $i\in [n]$ and
$i'\in[n+1,2n]$, and note that the constraints involving
$r_i$ and $r_{i'}$ are of four types: they belong to one of
$\mathcal{C}_H^{crcr}$, $\mathcal{C}_H^{crrc}$,
$\mathcal{C}_H^{rcrc}$, or $\mathcal{C}_H^{rccr}$.

Note that the satisfied constraints from $\mathcal{C}_H^{crcr}$
involving $r_i$ and $r_{i'}$ are exactly the constraints of the form
$(c_k,r_i,c_{k'},r_{i'})$ with $k\in [1,\phi(i)]$ and $k'\in
[n+1,\phi(i')]$. Thus there are $\deg(\mathcal{R}_i^{\le
  \phi(i)},\mathcal{R}_{i'}^{\le \phi(i')})$ such constraints, where
$\mathcal{R}_i^{\le \phi(i)}$ is the set of the $\phi(i)$ first vertices of
the $i^{th}$ row of $H$. Similar arguments on $\mathcal{C}_H^{crrc}$,
$\mathcal{C}_H^{rcrc}$, and $\mathcal{C}_H^{rccr}$ lead to the fact
that the total number of satisfied constraints involving $r_i$ and
$r_{i'}$ is exactly $\deg(\mathcal{R}_i^{\le
  \phi(i)},\mathcal{R}_{i'}^{\le \phi(i')}) + \deg(\mathcal{R}_i^{\le
  \phi(i)},\mathcal{R}_{i'}^{\ge \phi(i')}) + \deg(\mathcal{R}_i^{\ge
  \phi(i)},\mathcal{R}_{i'}^{\ge \phi(i')}) + \deg(\mathcal{R}_i^{\ge
  \phi(i)},\mathcal{R}_{i'}^{\le \phi(i')})$, which equals
$\deg(\mathcal{R}_i,\mathcal{R}_{i'}) +
\deg((i,\phi(i)),\mathcal{R}_{i'}) + \deg(\mathcal{R}_i,(i',\phi(i'))) +
\deg((i,\phi(i)),(i',\phi(i')))$, and which by condition (B) equals
$(n+2)\Delta^{i,i'} +\deg((i,\phi(i)),(i',\phi(i')))$. Summing this for
every $i\in [n]$ and $i'\in [n+1,2n]$ leads to the mentioned value.
\end{proof}

Thus, if an optimal ordering satisfies
${2Dn}\choose{2}$${2n+1}\choose{2}$$+(n+2)\sum\Delta^{i,i'} + n^2$
constraints, there is a convenient ordering (by Claim~\ref{l-sequ})
defining a function $\phi$ such that the graph $H[V_\phi]$ has $n^2$
edges (by Lemma~\ref{l-nb}). Since $H$ is bipartite, this subgraph is
bipartite and the number of edges implies that it is a $K_{n,n}$ (with
one vertex per row, by construction).  Conversely if $H$ has a
$K_{n,n}$ with one vertex per row , one can easily construct a
convenient ordering satisfing
${2Dn}\choose{2}$${2n+1}\choose{2}$$+(n+2)\sum\Delta^{i,i'} + n^2$
constraints.  This concludes the proof of the theorem.
\end{proofof}

Finally, up to our best knowledge, no $O^*(c^n)$-time algorithm for {\sc Permutation CSP} is known for $c<2$ even when the arity is two. Recently, Cygan et al. proved in \cite{CDL11} that for a number of basic problems including {\sc Hitting Set}, improving upon the $2^n$-barrier contradicts \texttt{Strong Exponential Time Hypothesis (SETH)} \cite{IP01}. Although SETH is not as widely believed as ETH, proving or disapproving SETH will be a major breakthrough in the domain. We leave it as an open problem whether it is possible to break the $2^n$-barrier for arity up to three.

\bibliographystyle{alpha}

\end{document}